\documentclass[12pt,a4paper]{amsart}
\usepackage[colorlinks=true,linkcolor=blue,urlcolor=blue,citecolor=blue]{hyperref}
\usepackage{amsmath, amsthm, amsfonts, amssymb, txfonts, mathrsfs, graphicx, bm,verbatim,siunitx}
\newcounter{algorithm}
\usepackage{breqn}
\theoremstyle{plain}
\newtheorem{theorem}{Theorem}
\newtheorem{lemma}{Lemma}[section]
\newtheorem{prop}{Proposition}[section]
\newtheorem{example}{Example}[section]
\newtheorem{corollary}[prop]{Corollary}
\newtheorem{definition}{Definition}
\newtheorem{algo}[algorithm]{Algorithm}
\newtheorem{remark}{Remark}
\newtheorem{conjecture}{Conjecture}[section]

\usepackage{algorithm}
\usepackage{algpseudocode}
\usepackage{algorithmicx}
\usepackage{cleveref}
\usepackage{caption}
\usepackage{enumitem}
\usepackage{listings}
\numberwithin{equation}{section}

\renewcommand{\phi}{\varphi}

\newcommand{\NN}{\mathbb{N}}

\newcommand{\ZZ}{\mathbb{Z}}
\newcommand{\QQ}{\mathbb{Q}}
\newcommand{\KK}{\mathbb{K}}
\newcommand{\FF}{\mathbb{F}}
\DeclareMathOperator{\ord}{ord}
\DeclareMathOperator{\sep}{Sep}

\newcommand{\Init}{\mathcal{I}}
\newcommand{\rsy}{R_{\boldsymbol{y}}}
\newcommand{\rs}{R_{\boldsymbol{s}}}
\newcommand{\mdt}{\mathrm{M}_{\delta_2}}
\newcommand{\mdtm}{\mathrm{m}_{\delta_2}}
\newcommand{\mdk}{\mathrm{M}_{\delta_k}}
\newcommand{\mtk}{\mathrm{M}_{\theta_k}}
\newcommand{\mtkm}{\mathrm{m}_{\theta_k}}
\newcommand{\rta}{r_{\theta}}
\newcommand{\sdr}{\mathscr{D}_{r_{\delta}^*}}
\newcommand{\mdkm}{\mathrm{m}_{\delta_k}}
\newcommand{\rd}{r_{\delta}^*}
\newcommand{\ld}{l_{\delta}^*}
\newcommand{\Dd}{D_{\delta}^*}
\newcommand{\rda}{r_{\delta}}
\newcommand{\p}{\mathscr{P}}
\newcommand{\fsig}{\mathbb{F}_{\sigma}}

\begin{document}
\title{D-Algebraic Guessing}
\author{Bertrand Teguia Tabuguia}
\address{University of Oxford, Wolfson Building, Parks Road, Oxford OX1 3QD, UK}
\email{email@bertrandteguia.com}
\begin{abstract}
Given finitely many consecutive terms of an infinite sequence, we discuss the construction of a polynomial difference equation that the sequence may satisfy. We also present a method to seek a candidate polynomial differential equation for its generating function. It appears that these methods often lead to effective D-algebraic operations.
\end{abstract}
\maketitle

\section{Introduction}\label{sec:intro}

The phrase {\it differentially-algebraic functions} refers to functions of a continuous variable that satisfy polynomial differential equations. Examples include the Weierstrass $\wp$ function, which satisfies the ODE $(\wp')^2=4\wp^3-g_2\wp-g_3$ for some complex constants  $g_2,g_3$, and Painlevé transcendents given by ODEs of the form $y''=R(y, y',t),y'=\frac{dy}{dt}$, where $R$ is a rational function. \textit{Difference-algebraic sequences}, on the other hand, satisfy polynomial difference equations. This includes all sequences defined by linear recurrence equations with constant coefficients, such as Fibonacci numbers $(F_n)_{n\in\NN}: F_{n+2}=F_{n+1}+F_n, F_1=1, F_0=0$; and as it turns out, sequences defined by linear recurrence equations with polynomial coefficients in the index variable \cite{teguia2024computing}. For instance, Catalan numbers $(C_n)_{n\in\NN}$ known by the recursion $(n+2)C_{n+1}=(4n+2)C_n,\,C_0=1,$ may be defined as 
\begin{small}
\[C_{n+2}=\frac{2 C_{n+1} \left(8 C_{n}+C_{n+1}\right)}{10 C_n-C_{n+1}},\, C_0=1,\, C_1=1.\]
\end{small}
For both objects, we use the adjective {\it D-algebraic}. The main goal behind D-algebraic guessing is to use some consecutive terms of an infinite sequence to conjecture that the sequence or its generating function is D-algebraic and provide an equation that certifies it. The conjecture may be proved from known algebraic properties of the sequence or its generating function. The following proposition is an example.

\begin{prop}\label{prop:catoFib} The sequence $(s_n)_{n\geq 1}, s_n\coloneqq C_n/F_n$, where $F_n$ and $C_n$ are as above, satisfies the following recurrence equation
\begin{small}
\begin{multline*}
-140\,s_{n+3}^2\,s_{n+1}^3\,s_{n} + 2\,s_{n+3}^2\,s_{n+1}^3\,s_{n+2} + 2240\,s_{n+3}^2\,s_{n+1}^2\,s_{n}^2 + 52\,s_{n+3}^2\,s_{n+1}^2\,s_{n}\,s_{n+2} \\
- s_{n+3}^2\,s_{n+1}^2\,s_{n+2}^2 + 1176\,s_{n+3}^2\,s_{n+1}\,s_{n}^2\,s_{n+2} - 27\,s_{n+3}^2\,s_{n+1}\,s_{n}\,s_{n+2}^2 - 140\,s_{n+3}^2\,s_{n}^2\,s_{n+2}^2 \\
+ 544\,s_{n+3}\,s_{n+1}^3\,s_{n}\,s_{n+2} + 2\,s_{n+3}\,s_{n+1}^3\,s_{n+2}^2 - 6912\,s_{n+3}\,s_{n+1}^2\,s_{n}^2\,s_{n+2}- 140\,s_{n+3}\,s_{n}^2\,s_{n+2}^3  \\
+ 4096\,s_{n+1}^2\,s_{n}^2\,s_{n+2}^2 - 332\,s_{n+3}\,s_{n+1}^2\,s_{n}\,s_{n+2}^2 - 2\,s_{n+3}\,s_{n+1}^2\,s_{n+2}^3 - 832\,s_{n+3}\,s_{n+1}\,s_{n}^2\,s_{n+2}^2\\
+ 512\,s_{n+1}\,s_{n}^2\,s_{n+2}^3 - 34\,s_{n+3}\,s_{n+1}\,s_{n}\,s_{n+2}^3- 512\,s_{n+1}^3\,s_{n}\,s_{n+2}^2 - 4\,s_{n+1}^3\,s_{n+2}^3\\
- 32\,s_{n+1}^2\,s_{n}\,s_{n+2}^3 = 0.
\end{multline*}
\end{small}
\end{prop}

As a multivariate polynomial, the left-hand side of the difference equation from Proposition \ref{prop:catoFib} is of total degree $6$. The highest index in the equation is $n+3$. We then say that the sequence $(s_n)_{n\geq 1}$ is D-algebraic of order $3$ and degree $6$. Note that it is not enough to have an equation satisfied by the sequence; it is essential that the equation enables the computation of any term from a certain index. This excludes equations like $s_{n+1}\,s_n=0$ for the power series coefficients of the sine and cosine functions. 

An easy proof of the proposition follows from evaluating the equation with a closed form of $C_n/F_n$. We give a proof independent of the explicit formulas of $C_n$ and $F_n$ in Example \ref{ex:arthmDseq}.

We comment that having a difference equation free of the index variable among its coefficients provides not only a view of sequences as points in the algebraic geometric sense, but also a constant recursive rule for generating terms of its sequence solutions. Indeed, each step of the calculation is not tied to its specific position in the sequence. This fact can be exploited to investigate fast evaluations of rational recursive sequences using finite fields. We also mention that these sequences arise naturally as outputs of some polynomial automata in computer science \cite[Example 1]{benedikt2017polynomial}.

The method of guessing has been extensively used in combinatorics to prove conjectures with D-finite functions. These are functions and sequences, such as the Catalan numbers, whose power series coefficients and generating functions, respectively, satisfy equations of the same description; hence, the choice of a single name for both. For instance, the generating function of Catalan numbers satisfies a second-order linear differential equation with polynomial coefficients of degree at most two in the independent variable. For further reading about these objects, we recommend the book \cite{kauers2023d} and the references therein. Using the guess-and-prove paradigm with D-finite closure properties, the following two theorems were established.

\begin{theorem}[Kauers, Koutschan, and Zeilberger \cite{kauers2009proof}]\label{theo:theo1} Let $g(n;i,j)$ denote the number of Gessel walks going in $n$ steps from $(0,0)$ to $(i,j)$. Then
\[g(2n+1;0,0)=0\,\quad \text{and}\,\quad g(2n;0,0)=\frac{(5/6)_n(1/2)_n}{(5/3)_n(2)_n}\,\, (n\in\NN),\]
where $(a)_n\coloneqq a(a+1)\cdots (a+n-1)$ is the Pochhammer symbol.
\end{theorem}

\begin{theorem}[Bostan and Kauers \cite{bostan2010complete}]\label{theo:theo2} The complete generating function for Gessel walks $G(t;x,y)=\sum_{i=0}^{\infty}\sum_{j,k\in\ZZ}g(i;j,k)\,t^ix^jy^k$ is algebraic.
\end{theorem}

As a consequence of the bridge between D-finite sequences and rational recursions established in \cite{teguia2024rational,teguia2024computing}, Theorem \ref{theo:theo1} might be analogously obtained by proving the required identities with the corresponding rational recursions. Note that non-guessing algebraic proofs of the above two theorems appeared in \cite{bousquet2016square}. A central motivation for considering nonlinear polynomial equations for sequences is to enlarge their closure properties beyond D-finiteness. For example, the sequence $(s_n)_{n\geq 1}$ from Proposition \ref{prop:catoFib} is not D-finite, implying the same conclusion for its generating function.  

The guess-and-prove paradigm for D-finite functions is presented in an artistic manner in \cite{yurkevich2022art}, featuring recent applications in number theory and biology. For a general illustration of the impact of computer algebra in combinatorics, we recommend \cite{bostan2017computer}. Corresponding software includes Gfun \cite{salvy1994gfun} for Maple and Guess \cite{kauers2009guessing} for Mathematica. A FriCAS implementation with a D-algebraic flavor was presented in \cite{hebisch2011extended}.

For D-algebraic functions, we allow the coefficients to be polynomials in the index variables. This makes D-finite functions naturally seen as D-algebraic functions of degree $1$, and algebraic functions as D-algebraic functions of order $0$. Recent results in \cite{bostan2020exponential}, \cite{bernardi2020counting}, \cite{bousquet2020generating}, and \cite{chen2025patterns} show a growing interest in D-algebraic functions solving nonlinear differential equations. A generalized family of weighted tree automata with D-algebraic generating functions is studied in \cite{manssour2024differential}. 

Technically, guessing algorithms for D-algebraic functions can always be regarded as particular instances of the Hermite-Padé approximation \cite{beckermann1994uniform}. As efficiency is not the primary focus of this paper, we simply note that the linear systems from these computations can be handled more effectively using Hermite-Padé approximation. Nevertheless, because it is almost unavoidable, we address computations over finite fields to avoid slowing the algorithm with long-integer arithmetic. We also mention the work in \cite{kauers2022guessing}, which combines guessing with the LLL algorithm to obtain results with fewer input terms, an aim we discuss from another algorithmic perspective in Example \ref{ex:littleDfun}.

The goal of this paper is thus justified in extending computer algebra capabilities beyond D-finiteness. We expect D-algebraic guessing to leverage the computational cost of D-algebraic operations \cite{van2019computing,RAB2024d,teguia2025arithmetic,teguia2024computing}. We present some examples in this direction. The idea is to guess a differential equation and verify that its associated differential polynomial belongs to the ideal of the operation (see Example \ref{ex:arthmDfun}). Alternatively, the guess can serve in finding the differential monomials of the target equation via computations over finite fields, and then an ansatz can be used to recover the correct equation with over the ground field (see Example \ref{ex:arthmDseq}). A similar strategy is exploited in \cite{mukhina2025projecting}.

The rest of the paper is organized into four main sections commencing with some preliminaries in Section \ref{sec:prelem}, then detailing our method for guessing D-algebraic functions in Section \ref{sec:dalgfun}, D-algebraic sequences in Section \ref{sec:dalgseq}, and presenting our software with some applications in Section \ref{sec:dalgapl}. Our Maple implementation is part of the NLDE package\footnote{available at \url{https://github.com/T3gu1a/D-algebraic-functions}.}.

\section{Preliminaries}\label{sec:prelem}

The methods described in this paper are conceptually simple and require only a relatively moderate theoretical background. A little knowledge of the basics of differential algebra and linear systems solving is enough to understand this paper.

We denote by $\KK$ a field of characteristic zero, typically a finite extension of the rationals.

\subsection{Settings for functions} Consider the differential field $\FF_x\!\coloneqq\!(\KK(x),\partial_x)$, where $\partial_x$ is the additive map that satisfies the Leibniz rule with $\partial_x(x)=1$ and $\partial_x(c)=0$, $c\in\KK$. We call $\partial_x$ the derivative operator. For any $f\in \FF_x$, we use the notations $\partial_x(f)= f',$ $\partial_x(\partial_x^j(f))=\partial_x^{j+1}(f)=f^{(j+1)},j\in\NN$ $(f^{(0)}=f)$ for its derivatives. By convension we set $f^{(-1)}=1.$

Let $y_1,\ldots,y_m$ be differential indeterminates. By $R_{\boldsymbol{y}}\coloneqq \FF_x\{y_1,\ldots,y_m\}$ we denote the ring of differential polynomials in $\boldsymbol{y}=y_1,\ldots,y_m$; i.e., $y_i^{(j)}\in R_{\boldsymbol{y}},$ for all $j\in\NN$ and $i=1,\ldots,m.$ The choice of $\KK(x)$ as base field is only for convenience of computational purposes; otherwise, we only deal with differential polynomials with coefficients in $\KK[x]$. 

A differential ideal $I\subset R_{\boldsymbol{y}}$ is an ideal such that $\partial_x(I)\subset I$. For a finite sequence $p_1,\ldots,p_{\ell}\in R_{\boldsymbol{y}},$ we denote by $<p_1,\ldots,p_{\ell}>$, the minimal differential ideal containing $p_1,\ldots,p_{\ell}$ and their derivatives.

Let $p\in R_{\boldsymbol{y}}.$ The order of $p$ with respect to (w.r.t.) $y_j$, denoted $\ord_{y_j}(p)$ is the maximum $r$ such that $y_j^{(r)}$ appears in $p$. The degree of $p$ w.r.t. $y_j$, denoted $\deg_{y_j}$, is its total degree as a multivariate polynomial in $\FF_x[y_j,\ldots,y_j^{(r)}],$ $r\!=~\!\ord_{y_j}(p).$ The separant of $p$ w.r.t. $y_j$ is defined as
\[\sep_{y_j,p} \coloneqq \frac{\partial p}{\partial y^{(r)}}\in\rsy,\,\, r=\ord_{y_j}(p).\]
Throughout the paper, we discuss multivariate differential polynomials in the context of ideals; our primary focus, however, is on the zeros of univariate differential polynomials. This corresponds to $m=1$ and $\rsy=R_{y_1}$ above, in which case $\ord_{y_1},\deg_{y_1}$, and $\sep_{y_1,p}$ are simply $\ord, \deg,$ and $\sep_p$, respectively. The degree of $C(x)\in\KK[x]$ is denoted $\deg_x(C(x)).$

\begin{definition}[D-algebraic function] A function $f$ is said to be differentially algebraic (or D-algebraic) over $\FF_x$ if there exists a differential polynomial $p\in R_y$ such that $p(f)\!\coloneqq\!p(y\!=\!f)\!=\!0$ and $\sep_p(f)\neq 0.$ In this case $p$ is called a defining differential polynomial for $f$. 
\end{definition}
Note that the definition of a D-algebraic function relies on the base field only for the differential polynomial and not for the initial values.

\begin{example}\label{ex:zeta0} As we are interested in generating functions, we identify any function with its power series representation around zero. The function
\[f\coloneqq \frac{1 - \pi \sqrt{x} \cot(\pi \sqrt{x})}{2x} = \frac{\pi^{2}}{6} +\frac{\pi^4}{90} x + \frac{\pi^6}{945} x^2 + \cdots=\sum_{n=0}^{\infty}\zeta(2n+2) x^{n},\]
is D-algebraic over $(\QQ(x),\partial_x)$ and can be defined with the differential polynomial $p=-2{y}^2 + 5{y'} - 4x{y'}{y} + 2x{y''}$. Indeed, we have $p(f)=0$, and $\sep_p=x$ implying $\sep_p(f)\neq 0$. 
\end{example}

For any $p\in R_y$, we call the equation $p=0$ the algebraic differential equation associated to $p$.

\subsection{Settings for sequences} We consider the difference field $\fsig\coloneqq (\KK,\sigma)$, where $\sigma$ is an endomporphism in $\KK.$ We write $\sigma^j(a)=a_j\in\KK, a\in\KK.$ The endomorphism $\sigma$ is extended to define $\rs\coloneqq \fsig\{s_{1,n},\ldots,s_{m,n}\},$ the ring of difference polynomials in the difference indeterminates $\boldsymbol{s}\!\coloneqq\! s_{1,n},\ldots,s_{1,m},$ where $\sigma^j(s_{i,n})=s_{i,n+j}.$ In an analogous manner as in the previous subsection, we define the notion of difference polynomial, algebraic difference equation, difference ideal, order, and degree.

Let $p\in\rs$. We say that $p$ is \textit{rationalizing} w.r.t. $s_{j,n}$ if it is linear in $s_{j,n+r}, r=\ord_{s_{j,n}}(p)$. The initial of $p$ w.r.t. $s_{j,n}$, denoted $\Init_{s_j,p}$, is the leading polynomial coefficient of $p$ viewed as a univariate polynomial coefficient in $s_{j,n+r}.$

\begin{definition}[D-algebraic sequence \cite{teguia2024computing}]\label{def:dalgseq} Let $p\in R_s$, and $p=p_0p_1\cdots p_{\ell},$ its algebraic decomposition into irreducible components. Let $u\coloneqq (u_n)_{n\in\NN}$ be a sequence and denote $\NN_{p_j}(u_n)\coloneqq \{n\in\NN,\, p_j(u_n)=0\}.$ The sequence $u$ is difference-algebraic (or D-algebraic) over $\fsig$ with defining difference polynomial $p$, if:
\begin{enumerate}
    \item For all $j\in\{0,\ldots,\ell\}$, $\#\NN_{p_j}(u_n)=\infty;$
    \item For all $j\in\{0,\ldots,\ell\}$, $\#\{n\in\NN_{p_j}(u_n),\, \Init_{p_j}(u_n)=0\}<\infty;$
    \item For all $n\in\NN$, there exists a unique $j\in\{0,\ldots,\ell\}$ such that $p_j(u_n)=~0$.
\end{enumerate}
\end{definition}

\begin{example} Let $p\coloneqq\left(2 s_ns_{n+1}+s_n+s_{n+1}\right)\left(2s_{n+1} s_n-s_n-s_{n+1}\right)\in R_s$ with $\KK=\QQ$. With $u_0\coloneqq 1$, the following recursion deduced from $p$ defined the sequence $u\coloneqq ((-1)^n/(2n+1))_{n\in\NN}$.
\[
    u_{n+1}= \begin{cases}
                     \frac{u\left(n\right)}{2u\left(n\right)-1},\, \text{if}\,\,\, n\,\,\, \text{is odd,}\\
                     -\frac{u\left(n\right)}{2u\left(n\right)+1},\, \text{if}\,\,\, n\,\,\, \text{is even.}
                     \end{cases}
\]
This shows that $u$ is D-algebraic. Moreover, using the algorithm from \cite{teguia2024rational}, one shows that $u$ can also be defined with the rationalizing second-order difference polynomial $(2s_n+s_{n+1})\,s_{n+2}+s_ns_{n+1}$.
\end{example}

\section{Differentially algebraic functions}\label{sec:dalgfun}

\subsection{Linear to quadratic differential equations}

 Jointly with Koepf in \cite{TeguiaDelta2}, the author studied quadratic differential equations to extend the family of hypergeometric-type power series. This later led to a guessing algorithm tailored to solutions of differential equations of degree at most $2$ \cite{TBguessing}. As an introduction to our general approach, we consider one of the main examples in \cite{TBguessing}, namely the recovery of the generating function from Example \ref{ex:zeta0}.
 
Let $S\coloneqq \{\zeta(2j+2),j=0,\ldots,N\}, N\in\NN.$ We seek a differential polynomial $p\in R_y\coloneqq \FF_x\{y\}$ of degree at most $2$ such that $p(\sum_{n=0}^{\infty} \zeta(2n+2)x^n)=0$ and $p(0)=0$ (homogeneity).

There are essentially two minor complexities to address to derive an algorithm that extends guessing from linear to quadratic differential equations. First, we ought to define an ordering for the differential monomials in
\[Y_2\coloneqq \{y^{(i)}y^{(j)},\,(i,j)\in(\NN\cup \{-1\})^2 \}\coloneqq \{1,y,y^2,y',y'y',{y'}^2,y''\ldots\}.\]
We consider monomials of lower order as the lowest, and if the order of two monomials coincide, i.e., if $\ord(m_i)=\ord(m_j)=r,$ $m_i,m_j\in Y_2$, then we compare $m_i/y^{(r)}$ and $m_j/y^{(r)}$. This order, which we call {\it (differential) graded quadratic lexicographic ordering}, is total on $Y_2$. Thus, we can uniquely map any nonnegative integer to an element of $Y_2\setminus \{1\}$. This is explicitly constructed in \cite{TeguiaDelta2}. For convenience, we define $\delta_2$ as the map
\[\delta_2: (y,j) \mapsto (j+1)\text{st lowest element in }Y_2\setminus \{1\}.\]
For a given $\delta_2(y,j),$ $j$ is called the $\delta_2$-order. For example, the first five $\delta_2$ derivatives of $\ln(x)$ are $\ln\! \left(x\right),\ln\! \left(x\right)^{2},1/x,\ln\! \left(x\right)/x,1/x^2,$ where $y$ is replaced by $\ln(x)$ in $Y_2$.

We may now want to know what the order of $\delta_2(y,j)$ is. It is straightforward to see that for every $r$, there are $r+2$ $\delta_2$-derivatives of order $r$. 
\begin{prop}\label{prop:mdt} The maximum $j$ such that $\ord(\delta_2(y,j))=r$ is 
$$\mdt(r)\coloneqq 1+r(r+5)/2=\binom{r+3}{2}-2.$$
\end{prop}
In \cite{TeguiaDelta2}, there is a difference of $2$ in the value of $\mdt(r)$ coming from counts involving the $\delta_2$-derivatives of order $-1$ ($\delta_2(y,-1)\!=\!1$), but the idea remains the same. Thus for $r>0$, the minimal $\delta_2$-derivative of order $r$ is $\mdtm(r)\coloneqq~\mdt(r-1)+1$. 

Let $r_{\delta}$ be the $\delta_2$-order of $p$, and $d$ be the maximum degree of its polynomial coefficients $C_j(x),j=0,\ldots,r$. We have $p=\sum_{j=0}^{r_{\delta}}C_j(x)\delta_2(y,j)$. Thus, we need at least $(r_{\delta}+1)(d+1)$ linear constraints on the constant coefficients in the $C(x)$'s to have enough evidence to claim the existence of $p$. 

The second minor complexity is to relate each term in the complete expansion of $p$ to a term in the recurrence equation satisfied by a power series solution of $p=0$. This is necessary to link the order of $p$ with that of its recurrence, and thereby deduce a relation between $N$ and $(r_{\delta}+1)(d+1)$.
\begin{prop}\label{prop:rewrule2} Let $f=\sum_{n=0}^{\infty}s_nx^n$ be a power series zero of $p\in R_y$. The term of the recurrence equation of $(s_n)_{n\in\NN}$ corresponding to the differential monomial $x^i\delta_2(y,j)=x^iy^{(j_1)}y^{(j_2)}$ is 
\[\begin{cases}
&(n+1-i)_{j_1} s_{n+j_1-i}\quad \text{if } j_2=-1,\\[0.1mm]
&\sum_{k=0}^{n-i} (k+1)_{j_1}(n-i-k+1)_{j_2} s_{k+j_1} s_{n-i-k+j_2},\,\, \text{if }\, (j_1+1)(j_2+1)\neq 0.
\end{cases}\]
$(a)_n\coloneqq a(a+1)\cdots (a+n-1)$ is the Pochhammer symbol or rising factorial.
\end{prop}
\begin{corollary}\label{prop:boundre} Let $p\in R_y$. If $\ord(p)=r$ and $p$ has polynomial coefficients of degree at most $d$, then the order of the recurrence equation of its power series zero after applying Proposition \ref{prop:rewrule2} is at most $r+d$.
\end{corollary}
Note that the bound in Proposition \ref{prop:boundre} is independent of $\delta_2$ and can therefore be used for the general algorithm in the next subsection. 

If $\ord(p)=r$, then using $S$ we can instantiate the recurrence equation for $n=0,\ldots,N-(r+d)$, thus obtaining $N-(r+d)+1$ equations with at least $(\mdtm(r)+1)(d+1)$ unknowns. Therefore choosing $\mdtm(r)$ and $d$ such that 
\begin{equation}\label{eq:d2cond}
    \left(\binom{r+2}{2}+1\right)(d+1)+r\leq N+2,
\end{equation}
yields an overdetermined linear system, which is not expected to have a nontrivial solution. If it does, then the corresponding $p$ is a candidate differential polynomial that vanishes at $\sum_{n=0}^{\infty}\zeta(2n+2)x^n$. It is interesting to notice that $\delta_2$ can be defined (or implemented) in such a way that \eqref{eq:d2cond} reads $\binom{r+2}{2}(d+1)+r\leq N+2.$

We consider an algorithm that iterates on the $\delta_2$ order of the $p$ sought. With $d=1$ and $N=14$, the maximum $r$ for \eqref{eq:d2cond} to hold is $r=2$ (we get an equality). Thus, the algorithm iterates from $r=0$ to $r=2$, corresponding to iterations from the $\delta_2$-order $r_{\delta}=0$ to $r_{\delta}=\mdtm(2)=5.$ At the last iteration, we have an ansatz equation of the form
\[(c_{0,0}+c_{0,1}x)y+(c_{1,0}+c_{1,1}x)y^2+\cdots+(c_{4,0}+c_{4,1}x)y'y+(c_{5,0}+c_{5,1}x)y''=0,\]
whose corresponding recurrence equation is given by
\begin{dmath*}
c_{0,0} s_n+c_{0,1} s_{n-1}+c_{1,0} \sum_{k=0}^n\! s_k s_{n-k}+c_{1,1} \sum_{k=0}^{n-1}\! s_k s_{n-k-1}+ \cdots +c_{4,0} \sum_{k=0}^n\! (k+1)(n-k+1) s_{k+1} s_{n-k+1}+c_{4,1} \sum_{k=0}^{n-1}\! (k+1)(n-k) s_{k+1} s_{n-k}+c_{5,0} (n+1) (n+2) s_{n+2}+c_{5,1} n (n+1) s_{n+1}=0.
\end{dmath*}
The latter is evaluated for $n=0\ldots11,$ the values of $s_n$ from $S,$ and $s_n=~0$ if $n<0$. Solving the resulting system, this algorithm predicts that all constant multiples of the differential polynomial in Example \ref{ex:zeta0} vanish at $\sum_{n=0}^{\infty}\zeta(2n+2)x^n$. Solving the associated differential equation with the initial values in $S$ concludes the recovery of the generating function. This is, of course, not a proof, but once we know a function whose initial series coefficients are those in $S,$ it becomes straightforward to verify its correctness.

\begin{remark}\label{remk:rmk1} Although the condition in \eqref{eq:d2cond} does not hold for $N=12$ with $r=2$ and $d=1$, $13$ initial terms suffices to recover the correct differential equation. The idea is to anticipate that $c_{5,0}=0$ is a possibility for the correct equation, and suppose that there are enough data for evaluation up to $n=11$. Thus, obtaining a linear system whose one of the equations is
\[\frac{236364091\pi^{24}}{201919571963756521875}c_{0,0}+\cdots+156s_{13}c_{5,0}+\frac{5263448\pi^{26}}{336196423516078125}c_{5,1}=0.\]
This is often exploited in our implementation. Note, however, that $13$ is not the minimum; as we shall see next, the same result can be obtained with $9$ terms only.
\end{remark}

\subsection{Bounded-degree ADEs}

Let $k\in\NN\setminus \{0\}$, and  $S\coloneqq\{s_{n_0},s_{n_0+1},\ldots,s_{N}\},$ $S\subset \KK,$ $n_0,N\in\NN,$ be a finite set of consecutive terms of a sequence $(s_n)_{n\in\NN}.$ Let $f$ be the generating function of $(s_n)_{n\in\NN}$. We want to find a candidate differential polynomial $p\in R_y$ of degree at most $k$ such that $p(f)=0$. Without loss of generality, we take $n_0=0$. For if $n_0>0$, the algorithm would seek $p$ such that $p(g)=0$, where $g=(f-\sum_{n=0}^{n_0-1}s_nx^n)/x^{n_0}$.

Let 
\begin{align*}
Y_k&\coloneqq \left\{\prod_{j=j_1}^{j_k} y^{(j)},\, (j_1,\ldots,j_k)\in(\NN\cup \{-1\})^k\right\}\\
   &= \left\{1,y,y^2,\ldots,y^k,y',y'y,\ldots,y'y^{k-1},{y'}^2y,\ldots\right\}.
\end{align*}
\begin{definition} Let $\succ_k$ be the order in $Y_k$ such that for any $t_1,t_2\in Y_k$, we have $t_1\succ_k t_2$ if and only if $\ord(t_1)\!\! > \!\! \ord(t_2)$; or $\ord(t_1/\gcd(t_1,t_2))\!\! > \!\! \ord(t_2/\gcd(t_1,t_2))$. We refer to this order as the (differential) graded degree-$k$ lexicographic ordering.
\end{definition}

Since any subset of $Y_k$ has a least element w.r.t. $\succ_k$, $(Y_k,\succ_k)$ is a well-ordered set. This makes the set $Y_k$ compatible with iterative algorithms.
\begin{definition} For any nonnegative integer $j$, we define the $n$th $\delta_k$-derivative of $y$ as the $(j+1)$st minimal element of $Y_k\setminus \{1\}$ w.r.t. $\succ_k.$
\end{definition}
It is straightforward to see that there are $\binom{r+k}{k+1}$ $\delta_k$-derivatives of order $r$. Hence, we can generalize Proposition \ref{prop:mdt} as follows.

\begin{prop}\label{prop:maxmind} Let $r\in\NN$, and $\mdk(r)$ and $\mdkm(r)$ be the maximum and the minimum integers $j$ such that $\ord(\delta_k(y,j))=r,$ respectively. We have $\mdkm(0)=0$, and
\[\begin{split}
  \mdk(r)  & \!\coloneqq\!  \binom{k+r+1}{k}-2,\\
  \mdkm(r) & \!\coloneqq\!  1+\mdk(r-1)\,\,\,\, (r>0).
\end{split}\]
\end{prop}

Using the convolution of series (or Cauchy product), we can establish the analogue of Proposition \ref{prop:rewrule2} for the conversion of differential equations into recurrence relations for their power series solutions. Since the explicit writing of these rewrite rules offers little insight to the algorithm, we choose to omit their presentation.

Using the same notation from the previous subsection, we write the ansatz $p=\sum_{j=0}^{r_{\delta}}C_j(x)\delta_k(y,j)$, where the $C(x)$'s are polynomial coefficients of degree at most $d\in\NN$.

\begin{lemma}\label{lem:dkcond} Let $p$ be an ansatz of order $r,$ degree $k$, and polynomial coefficients of degree at most $d$. If the linear system resulting from instantiating the recurrence equation of power series zeros of $p=0$ with the terms in $S$ is squared or overdetermined, then the following condition holds.
\begin{equation}\label{eq:dkcond}
    \left(\binom{r+k}{k}+1\right)(d+1)+r\leq N+2.
\end{equation}
\end{lemma}
We are now in a position to state the core of our algorithm. Note that the complete method is defined by its combination of Algorithm \ref{algo:Algo1}, and the discussions preceding its statement. In subsequent examples and applications, references to Algorithm \ref{algo:Algo2} should be understood to mean this complete method.

\begin{algorithm}[ht]\caption{Guessing D-algebraic functions.}\label{algo:Algo2}
    \begin{algorithmic} 
    \\ \Require $S=\{s_0,\ldots,s_N\}\subset \KK, d, k\in\NN.$
    \Ensure A basis of differential polynomials $p$ of degree at most $k$, with polynomial coefficients of degree at most $d$, such that $p\left(\sum_{n=0}^{N} s_n x^n\right) = 0$ and $p(0)=0$; or ``None'' (no ADEs found).
    \begin{enumerate}
    \item Let $r_{\max}$ be the maximum $r$ such that \eqref{eq:dkcond} holds, and $r_{\min}$ the starting order (usually $0$). 
    \item For $r_{\delta}=\mdk(r_{\min}),\mdk(r_{\min})+1,\ldots,\mdkm(r_{\max})$ do  
    \begin{enumerate}
        \item Let $p=\sum_{j=0}^{r_{\delta}}C_j(x)\delta_k(y,j)$, where $C_j(x)\in\KK[x]$ is of degree $d$ with undetermined coefficients, and let $r=\ord(p)$.
        \item Let $p_{\text{rec}}$ be the recurrence equation resulting from the conversion of $p=0$ for the power series solution $\sum_{n=0}^{\infty}s_nx^n$.


            \item Evaluate $p_{\text{rec}}$ for $n=0,\ldots,(d+1)(r_{\delta}+1)$ with the initial values in $\{s_{-r-d},\ldots,s_{-1}\}\cup S$, where $s_{-r-d}=\cdots=s_{-1}=0$. This produces an overdetermined linear system in the unknown coefficients from the $C(x)$'s.
            \item If the system has a nontrivial solution, verify that $p_{\text{rec}}$ evaluates to zero for $n=(d+1)(r_{\delta}+1)+1,\ldots, N-(r_{\max}+d)$. If it does, then stop and return $p$ with the unknown constant coefficients replaced by their corresponding value in the general vector solution.
        \end{enumerate}
    \item Return Algorithm \ref{algo:Algo1} $(S,d,k,\max\{r_{\min},r_{\max}\})$.
    \end{enumerate}	
    \end{algorithmic}
\end{algorithm}

The general vector solution mentioned in step (2)-(d) of Algorithm \ref{algo:Algo2} represents a vector space. Here, we consider the solution to be an arbitrary linear combination of a basis in that space. We also comment that the starting order $r_{\min}\leq r_{\max}$ may be optionally specified in the input of Algorithm \ref{algo:Algo2}. The algorithm would skip step (2) if $r_{\min}>r_{\max}$, and use $r_{\min}$ in step (3).

The objective of Algorithm \ref{algo:Algo1} is to use the data in $S$ to seek a differential polynomial with the same characteristic as Algorithm \ref{algo:Algo2}, but with a smaller number of unknowns and a $\delta_k$-order higher than $\mdkm(r_{\max})$. It follows from Lemma \ref{lem:dkcond} that there is an integer $\rd\in (\mdkm\!(r_{\max}),\,\mdkm\!(r_{\max}+1)]$ for which the linear system to be solved becomes underdetermined. We choose to determine $\rd$ algorithmically by examining the resulting linear system while incrementing the $\delta_k$-order.

Consider the ansatz $p=\sum_{j=0}^{\rd}C_j(x)\delta_k(y,j)$, such that $\deg_x(C_j(x))=d_j\in\NN.$ From the preceding steps of Algorithm \ref{algo:Algo2}, we know that no solution has been found for $\delta_k$-order $\rd-1$ with $d_j=d,$ $0\leq j\leq \rd$. Since $\rd$ yields an underdetermined system with $d_{\rd}=d$, it follows that there is $\ld\in[0,d)\cap \NN$ such that for $d_{\rd}=\ld$ and $d_j=d,$ $0\leq j\leq \rd$, the linear system is square. In Algorithm \ref{algo:Algo1}, we denote by $\Dd\coloneqq \rd d+\ld,$ the sum of degrees of polynomial coefficients corresponding to this square linear system, and consider the set
\begin{equation*}
    \sdr(d,N,r_{\max})\coloneqq \left\lbrace (d_0,\ldots,d_{\rd})\in\NN^{\rd+1},
    \begin{cases} 
    \max\{d_0,\ldots,d_{\rd}\}=d,\,\text{ and }\\ 
    \sum_{j=0}^{\rd}d_j=\Dd\leq N-(r_{\max}+d)
    \end{cases}
    \right\rbrace.
\end{equation*}

 This set includes all remaining cases where the data in $S$ yields a square or overdetermined linear system for the coefficients of a differential polynomial of $\delta_k$-order $\rd$ with polynomial coefficients of degree at most $d$. Indeed, from the definition of  $\Dd$, it follows that any tuple $t_d$ of polynomial degrees excluded from $\sdr(d,N,r_{\max})$ either results in an underdetermined linear system or defines a system whose solution space is contained within that of a system produced by a tuple in $\sdr(d,N,r_{\max})$.

 In general, we define $\sdr(d,N,r)$ for a given $r$ with corresponding $\rd$ and $\Dd$, as the set that contains all possible degrees of polynomial coefficients with maximum degree $d$ such that the ansatz $p$ leads to a square or overdetermined linear system with the data in $S$. Using classical results on integer partitions, one verifies that $\#\sdr(d,N,r) =\rd\binom{\Dd-d+\rd-2}{\rd-2}$.

\begin{algorithm}[ht]\caption{Guessing D-algebraic functions with a fixed order.}\label{algo:Algo1}
    \begin{algorithmic} 
    \\ \Require $S=\{s_0,\ldots,s_N\}\subset \KK, r, d, k\in\NN.$
    \Ensure A basis of differential polynomials $p$ of order $r$ and degree at most $k$, with polynomial coefficients of degree at most $d$, such that $p\left(\sum_{n=0}^{N} s_n x^n\right) = 0$ and $p(0)=0$; or ``None''.
    \begin{enumerate}
    \item Find $\rd\in[\mdkm(r),\, \mdkm\!(r+1)]\cap \NN$ such that the linear system in step $(2)-(c)$ of Algorithm \ref{algo:Algo2} is underdetermined. This may be done by pursuing the computations of step $(2)$ in Algorithm \ref{algo:Algo2} while incrementing the $\delta_k$-order $\rda$. If a solution is found, it would be returned accordingly; otherwise, $\rd$ is found and the algorithm proceeds.
    \item For $(d_0,\ldots,d_{\rd})\in\sdr(d,N,r)$ do
    \begin{itemize}
        \item Apply steps (a) to (d) from step $(2)$ of Algorithm \ref{algo:Algo2} using $p=~\sum_{j=0}^{\rd}C_j(x)\delta_k(y,j),$ with $\deg_x(C_j(x))=d_j$ and $r_{\max}=r$.
    \end{itemize}
    \item Return ``None''.
    \end{enumerate}	
    \end{algorithmic}
\end{algorithm}

\begin{remark} Due to the loop in step $(2)$, the runtime of Algorithm \ref{algo:Algo1} can be high. In practice, this step is implemented as an option. We may prefer to increase $N$ (find more data) rather than enabling its execution, unless there is no better option or the expected order and degree are relatively small.
\end{remark}

Algorithm \ref{algo:Algo1} is best exploited when a desired minimal order $r_{\min}> r_{\max}$ is specified, where $r_{\max}$ is defined in Algorithm \ref{algo:Algo2}. In that case, Algorithm \ref{algo:Algo1} is executed with $\rd=\mdkm(r_{\min})$. This makes Algorithm \ref{algo:Algo1} seek results that it would not try to obtain when called with orders $r\leq r_{\max}.$

Let us further analyze the guessing of the generating function from the previous subsection. The differential polynomial obtained is of $\delta_2$-order $5$. The polynomial coefficients are of degrees $-\infty$, $0,0,1$, $-\infty$, and $1$, listed in increasing order w.r.t. $\succ_2$, where $\deg_x(0)=-\infty$. Viewing the zero polynomial as a constant polynomial with a zero coefficient implies that it requires one unknown to be determined. In this sense, for the generating function of $(\zeta(2n+2))_{n\in\NN}$, the ``necessary'' data for its recovery is estimated to be $1+1+1+2+1+2=8.$ This is one less than the minimum amount of data required by our algorithm. Indeed, calling Algorithm \ref{algo:Algo2} with $d=1,k=2,r_{\min}=2,$ and $N+1=9$ initial terms, and exploiting the idea of Remark \ref{remk:rmk1} to suppose $s_9$ given, suffices to recover the correct differential polynomial. 

\noindent Nevertheless, in this example, we view the optimal bound for the amount of data required to recover the correct differential polynomial as $6$, the number of nonzero coefficients.  

\begin{definition} Let $f$ be a D-algebraic function defined by $p\in R_y$, $k=\deg(p),$ and $\rda$ be the $\delta_k$-order of $p$. Assuming $\deg_x(0)=-1,$ we define the optimal bound of the amount of data required to recover $f$ as $\sum_{j=0}^{\rda} (d_j+1),$ where $d_j$ is the degree of the polynomial coefficient in front of $\delta_k(y,j)$ in $p.$ We define the algorithmic bound with the same formula, assuming $\deg_x(0)=0.$
\end{definition}

\begin{definition} Suppose that $p_c\in R_y$ is a differential polynomial returned by Algorithm \ref{algo:Algo2} or Algorithm \ref{algo:Algo1}, where $\KK=\QQ(c_0,c_1,\ldots,c_{\ell})$. We can write 
\[p_c=\sum_{j=0}^{\ell} c_j p_j,\,\, p_j\in (\QQ(x),\partial_x)\{y\}.\]
We say that Algorithm \ref{algo:Algo2} or Algorithm \ref{algo:Algo1} recovers a D-algebraic function $f$ defined by a differential polynomial $p\in R_y$, if $p\in \langle p_0,\ldots, p_{\ell}\rangle.$
\end{definition}

\noindent We summarize the discussion of this section in the following two theorems.

\begin{theorem}\label{th:dkfinite} Algorithm \ref{algo:Algo2} is correct. Specifically, let $f$ be a D-algebraic function defined by a differential polynomial $p$ of degree $k$ and order $r$, with polynomial coefficients of degree at most $d$. Let $\tilde{f}=\sum_{n=0}^Ns_nx^n$ be a truncation of the power series of $f$. If Algorithm \ref{algo:Algo2} recovers $f$ from $\tilde{f}$, then the condition of Lemma \ref{lem:dkcond} holds.
\end{theorem}

\begin{theorem}\label{th:dkfinitedata} Algorithm \ref{algo:Algo1} is correct. Specifically, let $f$ be a D-algebraic function defined by a differential polynomial $p$ of degree $k$ and order $r$, with polynomial coefficients of degrees $d_0,\ldots,d_{\rda}\leq d$. Let $\tilde{f}=\sum_{n=0}^Ns_nx^n$ be a truncation of the power series of $f$. Algorithm \ref{algo:Algo1} recovers $f$ from $\tilde{f}$ if and only if $(d_0,\ldots,d_{\rda})\in\sdr(d,N,r)$.
\end{theorem}

\begin{algo}[Modular computations] We use the notation Algorithm $\!j$ $\!\!\!\mod \p$ to denote the variant of Algorithm $j$ that does the same computations as Algorithm $j$ but over the finite field of $\p$ elements. For Algorithm \ref{algo:Algo2} and \ref{algo:Algo1} this corresponds to taking $S$ modulo $\p$, evaluate the recurrence modulo $\p$, and solve the linear system modulo $\p$. These variants are crucial for doing guessing in practice. The idea is to use Algorithm \ref{algo:Algo2} $\!\!\!\mod \p_j$ for some primes $\p_1,\ldots,\p_m$, and use the Chinese remainder theorem to construct the result modulo $\prod_{j=1}^m\p_j$, which should be followed by rational reconstructions for the constant coefficients.
\end{algo}

\begin{example}[Arithmetic of functions with guessing]\label{ex:arthmDfun}
Consider the differential polynomials $p\coloneqq 2xy_1'+(x^2+1)y_1''$ and $q\coloneqq y_2+y_2''.$ We want to compute a differential polynomial $\rho$ verifying $\rho(f/g)=0$, for all generic $f,g$ such that $p(f)=q(g)=0$. Instead of computing for all $f,g$, the idea in this example is to take specific zeros of $p$ and $q$ by choosing explicit initial values. Here, we take $y_1(0)=0,y_1'(0)=1$ ($\arctan(x)$), and $y_2(0)=y_2'(0)=1$ ($\sin(x)+\cos(x)$). Thus, we can, in principle, have as many terms of the series of $f/g$ as we want.

With the first $64$ terms of the series, our implementation finds the following differential polynomial within five seconds.
\begin{dmath}
\rho\coloneqq(-3x^4 - 6x^2 - 3){z^{(2)}}^2 + (-x^4 - 8x^2 + 1){z}^2 + (-6x^4 - 24x^2 - 2){z'}^2 + (2x^4 + 4x^2 + 2){z^{(3)}}{z'} + (-4x^3 - 4x){z}{z'} + (4x^4 + 14x^2 + 2){z}{z^{(2)}} + (-6x^3 - 6x){z'}{z^{(2)}} + (2x^3 + 2x){z}{z^{(3)}}.
\end{dmath}
For the proof step, we verify that 
\begin{equation}
    \rho \in <p,q,z\,y_2-y_1>,
\end{equation}
which indeed holds. This verification can be performed within a second using differential reduction from the \texttt{DifferentialAlgebra} package \cite{DA}; see \cite{robertz2014formal} for more algorithmic details. Note that the standard Gr\"obner bases approach from \cite{RAB2024d,teguia2025arithmetic} could not terminate within half an hour. The success of the method here partly relies on the fact that ideal membership requires a lower computational cost than the full Gr\"obner basis computation required for this operation. Another explanation is that the choice of the initial values is sufficiently generic to ensure that the guessed differential polynomial vanishes at all $f/g$ when $p(f)=q(g)=0$.
\end{example}

In general, compared to the classical method, Algorithm \ref{algo:Algo2} requires fewer initial terms when enabled to try all polynomial coefficients of a given degree; however, this comes at a certain cost. All differential equations of D-algebraic functions from \cite{hebisch2011extended} are obtained with fewer data. The next example witnesses this observation.

\begin{example}[Guessing a function with necessary data]\label{ex:littleDfun}
Let $s_n\coloneqq \sum_{k=0}^n C_{3k}$, where $C_k$ is the $k$th Catalan number. With at least the first $22$ terms of $(s_n)_{n\in\NN}$, the current version (4.14, October 2025) of \texttt{Gfun} returns the following correct third-order linear differential equation with polynomial coefficients of degree at most $4$ for its generating function.
\begin{dmath}
(-48 + 1848x){y} + (6992x^2 - 1870x + 8){y'} + (4320x^3 - 2646x^2 + 27x){y''} + (576x^4 - 585x^3 + 9x^2){y^{(3)}} = 0.
\end{dmath}
Our algorithm finds the same differential equation with $15$ terms. The latter is also the minimal required by \texttt{Gfun} to find the corresponding D-finite recurrence equation. Note that the degrees of the polynomial coefficients are $1,2,3,4$, which yields the same algorithmic and optimal bound of $14$. 

The method from \cite{kauers2022guessing} can find the correct recurrence equation with only $8$ terms. Thus, while we can generally find ADEs with fewer terms than what the classical linear algebra method demands, in comparison with the amount of data required in \cite{kauers2009guessing}, we categorize Algorithm \ref{algo:Algo2} as succeeding with the ``necessary data.''
\end{example}

\section{Difference algebraic sequences}\label{sec:dalgseq}

For convenience of notation, we now use $u\coloneqq (u_n)_{n\in\NN}$ for the unknown infinite sequence, and use $s_n$ as the difference indeterminate. We want to recover $u$ from its consecutive terms $u_{n_0},u_{n_0+1},\ldots,u_N$. We seek a difference polynomial $p\in R_s$ such that $p(u)=0$. As in the previous section, we reduce the problem to finding a difference polynomial of a certain degree. Since all coefficients are constants, it is enough to design the analogue of Algorithm \ref{algo:Algo2} without Algorithm \ref{algo:Algo1}.

\subsection{Bounded-degree ADEs} 

Let $k\in\NN\setminus \{0\}$ and $U\coloneqq\{u_{n_0},u_{n_0+1},\ldots,u_N\},$  $U\subset\KK.$ Again, we take $n_0=0$. If $n_0>0$, the algorithm would seek $p\in R_s$ such that $p(v)=0$, where $v=(u_n)_{n\geq n_0}$, which is a sequence of the same {\it germ} as $u$ (they differ by finitely many terms).

Let 

\begin{align*}
\mathcal{S}_k &\coloneqq \left\{\prod_{j=j_1}^{j_k} \sigma^j(s_n),\, (j_1,\ldots,j_k)\in(\NN\cup \{-1\})^k\right\}\\
              &= \left\{1,s_n,s_n^2,\ldots,s_n^k,s_{n+1},s_{n+1}s_n,\ldots,s_{n+1}s_n^{k-1},{s_{n+1}}^2s_n,\ldots\right\},
\end{align*}
be the analogue of $Y_k$ from the differential case. In a similar manner, we define the ordering $\succ_k$ on $\mathcal{S}_k$.

\begin{definition} For any nonnegative integer $n$, we define the $n$th $\theta_k$-shift of $s_n$ as the $(n+1)$st minimal element of $\mathcal{S}_k\setminus \{1\}$ w.r.t. $\succ_k.$
\end{definition}

\begin{prop} Let $r\in\NN$, and $\mtk(r)$ and $\mtkm(r)$ be the maximum and the minimum integers $j$ such that $\ord(\theta_k(s_n,j))=r.$ We have $\mtk(r)=\mdk(r)$ and $\mtkm(r)=\mdkm(r)$, where $\mdk(r)$ and $\mdkm(r)$ are as in Proposition \ref{prop:maxmind}.
\end{prop}

We can evaluate an ansatz of $p$ of order $r$ for $n=0,\ldots,N-r$, thus obtaining $N-r+1$ equations with at least $\mtkm(r)+1$ unknowns.

\begin{lemma}\label{lem:tkcond} Let $p$ be an ansatz of order $r$. A square or overdetermined linear system arises from instantiating $p=0$ with the terms in $U$ if the following condition is met.
\begin{equation}\label{eq:tkcond}
    \binom{r+k}{k}+r\leq N+1.
\end{equation}
\end{lemma}

We can relate the condition of Lemma \ref{lem:tkcond} to the one from \cite[Section 1.5]{kauers2023d} and \cite{kauers2022guessing} concerning D-finite sequences. It reads $(r+1)(d+2)\leq N+2$ where $d$ is the maximum degree of the polynomial coefficients. Assuming $d=0$, i.e., linear recurrence with constant coefficients, the inequality becomes $2r\leq N$. Translating \eqref{eq:tkcond} for such equations is the same as taking $k=1$, which yields exactly the same inequality.

\begin{algorithm}[ht]\caption{Guessing D-algebraic sequences.}\label{algo:Algo3}
    \begin{algorithmic} 
    \\ \Require $U\coloneqq \{u_0,\ldots,u_N\}\subset\KK, k\in\NN.$
    \Ensure A basis of difference polynomials $p$ of degree at most $k$ such that $p\left(s_n\right) = 0,$ for all $n\leq N-\ord(p);$ or ``None''.
    \begin{enumerate}
    \item Let $r_{\max}$ be the maximum $r$ such that \eqref{eq:tkcond} holds.
    \item For $\rta=\mtk(0),\mtk(0)+1,\ldots,\mtkm(r_{\max})$ do  
    \begin{enumerate}
        \item Let $p=\sum_{j=0}^{\rta}C_j\theta_k(s_n,j)$, where $C_j\in\KK$ is of degree $d$ with undetermined coefficients, and set $r=\ord(p)$.
        \item Evaluate $p=0$ for $n=0,\ldots,(d+1)(\rta+1)$ with the initial values in $U$. This produces an overdetermined linear system in the unknowns $C$'s.
        \item If the system has a nontrivial solution, verify that $p$ evaluates to zero for $n=(d+1)(\rta.+1)+1,\ldots, N-(r_{\max}+d)$. If it does, then stop and return $p$ with the unknown constant coefficients replaced by their corresponding value in the general vector solution.
    \end{enumerate}
    \item Return ``None''.
    \end{enumerate}
    \end{algorithmic}
\end{algorithm}
\noindent We summarize our results in the following theorem.

\begin{theorem} Algorithms \ref{algo:Algo3} is correct. Specifically, let $u\coloneqq (u_n)_{n\in\NN}$ be a D-algebraic sequence defined by a difference polynomial $p$ of degree $k$ with $\theta_k$-order $\rta$, and order $r$. Let $\tilde{u}=(u_0,\ldots,u_N)$ be a truncation of $u$. Algorithm \ref{algo:Algo3} recovers $u$ from $\tilde{u}$ if and only if $\rta + r\leq N$. In particular, if Algorithm \ref{algo:Algo3} recovers $u$ from $\tilde{u}$ then $\binom{r+k}{k}+r\leq N+1.$ 
\end{theorem}

\begin{algo}[Modular computations] As in the differential case we denote Algorithm \ref{algo:Algo3} $\!\!\!\!\!\mod \p$ the variant of Algorithm \ref{algo:Algo3} that does the computation over the field of $\p$ elements, where $\p$ is a prime.
\end{algo}

\begin{example}[Arithmetic of sequences with guessing]\label{ex:arthmDseq} We aim to construct the difference polynomial in Proposition \ref{prop:catoFib} from the difference polynomials defining $(C_n)_{n\in\NN}$ and $(F_n)_{n\in\NN}.$ The method we present can be regarded as a variant of that discussed in Example \ref{ex:arthmDfun}. The main steps are as follows. We assume that we have two input difference polynomials of orders $r_1$ and $r_2$, with difference indeterminates $a_n$ and $b_n$, respectively.
\begin{enumerate}
    \item Find prime numbers $\p_1,\ldots,\p_m$ such that for any $j=0,\ldots,m$, Algorithm \ref{algo:Algo3}\! $\!\mod \p_j$ yields difference polynomials with the same difference monomials. In other words, the nonzero coefficients (their numerators) in the hypothetical difference polynomial are not divisible by these primes, and their product divides the remaining coefficients, which we claim to be zeros.
    \item Make an ansatz with undetermined coefficients whose monomials are the ones from step (1).
    \item Substitute the indeterminate sequence, say $s_n$, in that ansatz with the operation $\mathcal{O}(a_n,b_n),$ where $a_n$ and $b_n$ symbolically represent the D-algebraic zeros of the input difference polynomials. For this example, we have $\mathcal{O}(C_n,F_n)=C_n/F_n.$
    \item Use substitutions with the input difference polynomials to rewrite all occurrences of $a_{n+d_1}, d_1\geq r_1$ and $b_{n+d_2}, d_2\geq r_2$ with lower order terms, and take the numerator of the resulting expression. We look at this numerator as a multivariate polynomial in $a_n,\ldots,a_{n+r_1-1},$ $b_n,\ldots,b_{n+r_2-1}$.
    \item Solve the linear system obtained by equating the coefficients of this polynomial to zero.
\end{enumerate}
With $175$ terms and the primes $751$ and $5003$, we find the same set of difference monomials. The ansatz yields an overdetermined linear system of $41$ equations with $21$ unknowns. The solutions are constant multiples of the vector of coefficients in Proposition \ref{prop:catoFib}, concluding the proof. 
\end{example}

\begin{example}[Fibonacci at index $2^n$]\label{ex:littleDseq}
Using the first $15$ terms of $(F_{2^n})_{n\in\NN}$, where $F_n$ is the $n$th Fibonacci number, we find the following difference polynomials modulo $\p_1\coloneqq 101$ and $\p_2\coloneqq 103$:
\[\begin{split}
    &96s_n^5 + 5s_n^4 + 97s_n^3 + s_{n+1}^2s_n + 4s_n^2 + 100s_{n+1}^2\quad (\!\!\!\!\!\mod \p_1),\\
    &98s_n^5 + 5s_n^4 + 99s_n^3 + s_{n+1}^2s_n + 4s_n^2 + 102s_{n+1}^2\quad (\!\!\!\!\!\mod \p_2).
\end{split}\]
After reconstruction via Chinese remaindering modulo $\mu\coloneqq \p_1\p_2=10403$, we get
\begin{equation}\label{eq:fib2ton}
    -(s_n - 1)(5s_n^4 + 4s_n^2 - s_{n+1}^2),
\end{equation}
whose one of the factors yields the classical identity $F_{2^{n+1}}^2 = F_{2^n}^2\left(5F_{2^n}^2+4\right)$. The other factor implies that this identity is not valid for $n=0$. Note that although we obtained the correct result after applying the Chinese remainder theorem, it is always required to apply rational reconstructions modulo $\mu$, which here keeps \eqref{eq:fib2ton} unchanged.
\end{example}

\section{Software and applications}\label{sec:dalgapl}

The following link leads to the repository of the \texttt{NLDE} package in which all the algorithms of this paper have been implemented by the author.

\begin{center}
    \begin{small}
        \url{https://github.com/T3gu1a/D-algebraic-functions}
    \end{small}
\end{center}

We describe the corresponding commands and illustrate their use on some examples.

\subsection{Implementation} We implemented Algorithm \ref{algo:Algo2} and Algorithm \ref{algo:Algo3} and their variants for finite fields in our Maple \texttt{NLDE} package.
\begin{enumerate}
    \item For guessing D-algebraic functions over fields of characteristic zero, the command is as follows.
    \begin{small}
    \[
    \begin{split}
    &\texttt{> NLDE:-DalgFunGuess(L, degADE=2, degPoly=2,}\\
    &\texttt{startfromord=0, allPolyDeg=false);}
    \end{split}
    \]    
    \end{small}
    $\texttt{L}$ is the list of initial terms; $\texttt{degADE},$ with default value $2$, is the degree of the differential polynomial sought; $\texttt{degPoly}$, default value $2$, is the maximum degree of the polynomial coefficients. Two other interesting arguments are $\texttt{startfromord}$, with default value $0$, representing the starting order for the search. This is the $r_{\min}$ discussed after Algorithm \ref{algo:Algo2}; \texttt{allPolyDeg} is a boolean variable, which, when assigned $\texttt{true}$, enables Algorithm \ref{algo:Algo1} to be executed to fully exploit the data to search for a differential polynomial with polynomial coefficients of degree at most \texttt{degPoly} where Algorithm \ref{algo:Algo2} cannot.
    \item The variant of $\texttt{NLDE:-DalgFunGuess}$ for computations over finite fields. 
    \begin{small}
    \[
    \begin{split}
    &\texttt{> NLDE:-modDalgFunGuess(L, degADE=2, degPoly=2,}\\
    &\texttt{startfromord=0,allPolyDeg=false, modulus=7);}
    \end{split}
    \]        
    \end{small}
    The argument $\texttt{modulus}$ is a prime number for the characteristic of the field.
    \item For guessing D-algebraic sequences, one first needs to load the subpackage for sequences in Maple.
    \begin{small}
    \[\texttt{>  with(NLDE,DalgSeq):}\]    
    \end{small}
    Once loaded, the commands are accessible with the \texttt{DalgSeq} prefix. For computations over fields of characteristic zero, we have
    \begin{small}
    \[\texttt{> DalgSeq:-DalgGuess(L, degADE=2, startfromord=0);}\]
    \end{small}
    where the arguments are the analogues of the differential case for difference equations.
    \item For guessing recurrences over finite fields, the command is
    \begin{small}
    \[
    \begin{split}
    &\texttt{> DalgSeq:-modDalgGuess(L, degADE=2,}\\
    &\texttt{startfromord=0, modulus=7);}
    \end{split}
    \]        
    \end{small}
\end{enumerate}
\begin{remark}[Challenges] While experimenting with our implementation of Algorithm \ref{algo:Algo2}, we noticed that evaluating the recurrence of the ansatz to obtain linear equations becomes more time-consuming when the degree of the ansatz grows. For instance, let us consider the recurrence of the ansatz of $\delta_4$-order $\mathrm{m}_{\delta_4}=14$ (degree $4$) with polynomial coefficients of degree $8$. So there are $9\times 15=135$ unknowns. Let us find the coefficients of one of the unknowns by evaluating the following expression for $n=50,\ldots,100.$ 
\[\sum_{k_2=0}^{n-2}\sum_{k_1=0}^{k_2}\sum_{k_0=0}^{k_1}(k_0+1)(k_2-k_1+1)(n-1-k_2)s_{k_0+1}s_{k_1-k_0+1}s_{k_2-k_1+1}s_{n-1-k_2}.\]
In Maple, this may be effectively done as follows.
\begin{small}
\[\texttt{> seq(eval(T,[n=j,Sum=add]),j=50..100):}\]
\end{small}
Here $\texttt{T}$ represents the expression we want to evaluate. Note that this is before substituting initial terms, which can easily be combined with the command above using the $\texttt{subs}$ command. This evaluation alone takes about $9$ seconds in CPU Time on our working computer. Thus, taking into account all other terms, it may take about $15$ minutes to instantiate the full recurrence and get the linear system, which can be solved very effectively. A better implementation of this instantiation is essential to exploit the full capacity of our algorithm. Our current implementation uses the above Maple code, but we expect improvement in future versions.
\end{remark}

\begin{remark}[D-algebraic guessing as "Fishing"]
We can view D-algebraic guessing as a form of intellectual fishing. The domain of our search -- the ``fishing grounds''-- is the immense sea of sequences arising in other sciences such as combinatorics, computer science, and physics. The ``fish'' we seek are the underlying differential and difference equations that govern these sequences. Our pursuit is a form of knowledge translated into algorithms designed to methodically search for these candidate equations. The use of our software may also require a methodical approach. The user must intelligently explore various degrees for polynomial coefficients and equations and critically interpret the recurring algebraic structures revealed through computations over finite fields from the same input data. This philosophy underlies the design of our software's commands.
\end{remark}

\subsection{OEIS A189281} The $n$th term of the integer sequence \href{https://oeis.org/A189281}{A189281} from the online encyclopedia of integer sequences \cite{sloane2003line}, counts the number of permutations $\pi$ of the symmetric group $S_n$ such that $\pi(i+2)-\pi(i)\neq 2$ for $i=1,\ldots,n-2$. A conjecture from \cite{kauers2022guessing} states that this sequence satisfies a linear recurrence equation of order $8$ and degree $11$. We have been able to find the same recurrence equation from a guessed differential equation for the corresponding generating function. 

There are currently $301$ terms of the sequence on the OEIS. We did not apply our algorithm with fewer initial terms; instead, we downloaded all the data and saved it in the file \texttt{A189281.txt} in the same directory as the Maple worksheet. The file contains two columns of integers: the first column on the left is for indices, and the other is for their corresponding terms. We load these terms into a list in Maple via the following command.
\begin{small}
\[
\begin{split}
&\texttt{> L:= map(indexterms -> indexterms[2],}\\
& \texttt{readata("A189281.txt",2,integer):}
\end{split}
\]
\end{small}
Next, we use Algorithm \ref{algo:Algo2} $\!\!\!\!\mod \p$ for $\p \in\{503, 563, 571, 577, 587\}$ to search for a linear differential equation, i.e., $k=1$ in the algorithm, and \texttt{degADE=1} for the software, with polynomial coefficients of degree at most $11$ ($d=11$, \texttt{degPoly=11}). The code for the computation is as follows.
\begin{small}
\[
\begin{split}
&\texttt{> NLDE:-modDalgFunGuess(L, degADE=1,}\\
&\texttt{degPoly=11, startfromord=19, modulus=}\p\texttt{);}
\end{split}
\]
\end{small}
Each of these computations takes about $3$ to $5$ seconds. After noticing that a probable guess of order $19$ was found with one of the primes, we use the value $19$ as the starting order for the others. We did try other prime values and polynomial degrees. For instance, we also observed that the sequence may satisfy a differential equation of order $4$ with polynomial coefficients of at most $24$. However, as the space of polynomial coefficients obtained was two-dimensional, we did not consider it further.

Since the terms of the sequence are quite long integers (the term of index $50$ has $65$ digits), it is not expected that the chosen primes would lead to the correct equation over the integers via the Chinese remainder theorem. However, we noticed that for each of these primes, the returned differential equation has the same set of appearing differential monomials (including $x$). In total, there are $155$ terms in the expanded differential equation. Other primes, such as $569$, were leading to $154$ terms.  The good news here is that there are more data than monomials, and we can define an ansatz with unknown coefficients in front of these differential monomials to construct an overdetermined linear system. We only used $160$ terms of the sequence for it. We obtained a linear differential equation of order $19$ with polynomial coefficients of degree at most $11$. A more precise characterization of the guess is provided in the conjecture below. The recurrence equation of this differential equation is of order $16$. Executing van Hoeij's implementation \texttt{MinimalRecurrence} from \texttt{LREtools} on this recurrence yields exactly the conjectured recurrence from Manuel Kauers and Christoph Koutschan \cite[Conjecture 10]{kauers2022guessing}, which is printed on the OEIS. We mention that they obtained that equation with only $36$ initial terms.

We were unable to factorize the differential equation in Maple; however, the particular pattern of the degrees of its polynomial coefficients is easy to recall. The minimal order is $2$, and for $j=2,\ldots,11$ the polynomial coefficient in from of $y^{(j)}$ is of degree $j$. The remaining coefficients are of degree $11$.

\begin{conjecture} The generating function of the sequence \href{https://oeis.org/A189281}{A189281} satisfies a linear differential equation of order $19$ with polynomial coefficients of degrees 
\[11,11,11,11,11,11,11,11,11,10,9,8,7,6,5,4,3,2,-\infty,-\infty,\]
ordered in decreasing order from the coefficients of $y^{(19)}$ to that of $y^{(0)}=~y$. With the leading and the trailing terms made explicit, the differential equation has the form
\[
\begin{split}
&-64x^8(x - 1)(x + 1)^2\,y^{(19)}+\cdots\\
& \cdots+ \left(906154905600 x^{2}-3621295278720 x-504819584640\right) y^{(2)}=0.
\end{split}
\]
\end{conjecture}

Note that \texttt{Gfun} (version 4.14, \cite{salvy1994gfun}) finds a D-finite differential equation of order $4$ with polynomial coefficients of degree at most $24$, whose singularities include the ones in the conjecture. This suggests that the other singularities are what Mark van Hoeij refers to as apparent singularities.

\subsection{Catalan at $4n$} It was shown in \cite{teguia2024computing} that D-algebraicity for sequences is closed under composition with arithmetic progressions. The paper states further that the order of the subsequence thus defined remains the same, and gives an example with the sequence $(C_{3n})_{n\in\NN},$  where $C_n$ denotes the $n$th term of the Catalan sequence. We use guessing to compute such an equation for $(C_{4n})_{n\in\NN}$. According to Proposition 5.2 in \cite{teguia2024computing}, a second-order algebraic difference polynomial exists. We use modular arithmetic and the explicit formula of $C_{4n}$ to compute as many terms as necessary. For many primes, it is possible to generate $300$ terms within seconds. We obtain the shape of the plausible ADE sought with most primes. Let $\p$ (eg. $3697$) be a sufficiently large prime for which our implementation yields the shape below.

\vspace{-0.35cm}

\begin{small}
\[
\begin{split}
&\texttt{> with(NLDE:-DalgSeq):}\\
&\texttt{> modDalgGuess(L,degADE=2,startfromord=8,revars=s(n),modulus=}\p\texttt{);}
\end{split}
\]

\vspace{-0.35cm}

\begin{dmath*}
c_{21}s_{n+1}^8 + c_{2}s_{n+2}s_{n+1}^3s_{n}^4 + c_{13}s_{n+2}^2s_{n+1}^4s_{n}^2 + c_{14}s_{n+2}^3s_{n+1}^3s_{n}^2 + c_{15}s_{n+2}^4s_{n+1}^2s_{n}^2 + c_{17}s_{n+2}s_{n+1}^6s_{n} + c_{18}s_{n+2}^2s_{n+1}^5s_{n} + c_{19}s_{n+2}^3s_{n+1}^4s_{n} + c_{20}s_{n+2}^4s_{n+1}^3s_{n} + c_{3}s_{n+2}^2s_{n+1}^2s_{n}^4 + c_{4}s_{n+2}^3s_{n+1}s_{n}^4 + c_{7}s_{n+2}s_{n+1}^4s_{n}^3 + c_{8}s_{n+2}^2s_{n+1}^3s_{n}^3 + c_{9}s_{n+2}^3s_{n+1}^2s_{n}^3 + c_{10}s_{n+2}^4s_{n+1}s_{n}^3 + c_{12}s_{n+2}s_{n+1}^5s_{n}^2 + c_{1}s_{n+1}^4s_{n}^4 + c_{22}s_{n+2}s_{n+1}^7 + c_{25}s_{n+2}^4s_{n+1}^4 + c_{11}s_{n+1}^6s_{n}^2 + c_{6}s_{n+1}^5s_{n}^3 + c_{23}s_{n+2}^2s_{n+1}^6 + c_{5}s_{n+2}^4s_{n}^4 + c_{16}s_{n+1}^7s_{n} + c_{24}s_{n+2}^3s_{n+1}^5=0.
\end{dmath*}
\end{small}
The constants $c_i\in\ZZ/\p\ZZ,i=1\ldots25.$ We can recover the correct equation by solving a linear system with the first $27$ terms (in $\NN$) of $(C_{4n})_{n\in\NN}$, or by using sufficiently many primes so that the Chinese remainder theorem can be applied to recover it. The vector space of coefficients over $\QQ$ is the following.
\begin{small}
\begin{equation*}
\begin{split}
&(c_1,\ldots,c_{25})^T\in \QQ v,\\
&v=\Big(18446744073709551616,-\frac{49886373072382984192}{35},\frac{1326172548227923968}{35},\\
&-\frac{2097317242994688}{5},1653603968016 ,-31525197391593472,\\
&\frac{84896797473898496}{35},-\frac{370029665189888}{5},\frac{38327015208768}{35},-\frac{32002521408}{5},\\
&13159779794944,-\frac{7345572675584}{7},\frac{343283112608}{7},-\frac{5646204608}{5},\frac{308773608}{35},\\
&264241152 ,\frac{238821056}{5},-\frac{112084544}{7},\frac{19766576}{35},-\frac{177232}{35},\\
&-388080 ,4032 ,3064 ,-112 ,1\Big)^T
\end{split}
\end{equation*}  
\end{small}
The resulting equation is satisfied by the closed form of $C_{4n}$, proving that the equation is correct. We observed from our computations that $(C_{dn})_{n\in\NN}$ satisfies an ADE of degree $2d$ for $d\in\{1,2,3,4,5\}.$ We expect that the proof of the following conjecture may follow from the dynamical system approach discussed in \cite{teguia2024computing}.

\begin{conjecture} $(C_{dn})_{n\in\NN}$, $d\in\NN\!\setminus\!\{0\},$ is D-algebraic of order $2$ degree $2d$.
\end{conjecture}

\subsection{Differential tree automata \cite{manssour2024differential}} In this part, we want to construct the nonlinear differential equation of the generating function arising from the two-dimensional differential tree automaton $\mathcal{A}=(2,\mu)$ over the alphabet $\Sigma=\{\sigma_0,\sigma_1,\sigma_2\}$. The transition function $\mu$ is defined as $\mu(\sigma_0)=[0\quad 1],$
\[\mu(\sigma_1)=\begin{bmatrix}0&0\\1&0\end{bmatrix},\quad \text{and}\quad \mu(\sigma_2)=\begin{bmatrix}\boldsymbol{0}_{3\times1}&\boldsymbol{0}_{3\times1}\\1&\frac{x_1+1}{x_0}\end{bmatrix}.\] 
This automaton generates the sequence of labelled rooted trees. An equivalent construction is discussed in \cite[Example 2]{manssour2024differential}.
The transition induces a function $\boldsymbol{\mu}\colon T_{\Sigma}\rightarrow \KK^{1\times 2}$ such that for $t=\sigma_k(t_1,\ldots,t_k)$
\[\boldsymbol{\mu}(t)\coloneqq (\boldsymbol{\mu}(t_1)\otimes \cdots \otimes \boldsymbol{\mu}(t_k))\cdot \mu(\sigma_k)(|t|,|t_1|,\ldots,|t_k|),\]
where $|t|$ is the size of $t$, and $\otimes$ is the Kronecker product. We are interested in the formal series of general term
\[\boldsymbol{a}_n\coloneqq \left[a_n^{(1)}\quad a_n^{(2)}\right]\! \coloneqq\!\! \sum_{t\in T_{\Sigma},\, |t|=n+1} \boldsymbol{\mu}(t).\]
This leads to the recursion 
\[\left[a_{n+1}^{(1)}\quad a_{n+1}^{(2)}\right] = \left[a_n^{(2)}\quad \frac{1}{n}\sum_{k=0}^{n-1}(k+1)\,a_{k+1}^{(2)}a_{n-k}^{(2)}\right],\]
yielding the sequence $u_0=0,\, u_1=1,\, u_{n+1}=\sum_{k=0}^{n-1}(k+1)u_{k+1}u_{n-k}/n,\, n\geq 1.$ Putting the first $9$ terms of the sequence thus defined in $\texttt{L}$ is enough for the following code to guess the correct differential equation.
\begin{small}
\begin{verbatim}
> NLDE:-DalgFunGuess(L,degADE=2,allPolyDeg=true,startfromord=1)
\end{verbatim}
\end{small}
\[y - xy' + xy'y=0.\]
The corresponding generating function is $-W(-x)=\sum_{n=1}^{\infty}n^{n-1}x^n/n!$, where $W$ is the Lambert function. This is exactly the known exponential generating function of labelled rooted trees. 

In general, any D-algebraic guesser can be used to guess algebraic differential equations for generating functions arising from differential tree automata. The correctness of the equations found can be verified by reconstructing the tree automata using the corresponding recurrence equations.

\subsection{Primes and D-algebraicity} As a final application, let us consider the sequence of odd-indexed primes \href{https://oeis.org/A031215}{A031215} (even-indexed if starting at index $1$) whose first five terms are $3,7,13,$ $19,29.$ There are undoubtedly several qualitative arguments to explain why this sequence is not D-algebraic. We prove this by investigating the consequences related to the shape of the difference polynomials returned by our implementation. Using the first $55$ terms of the sequence, we get a first-order difference polynomial of degree $8$. This is obtained with the following command.
\begin{small}
\[
\texttt{> DalgSeq:-DalgGuess(L,degADE=8,startfromord=1,revars=s(n));}
\]
\end{small}
The possible coefficients form a two-dimensional vector space. We select the equation generated by one of the basis vectors, resulting in the equation below.
\begin{small}
\begin{equation}\label{eq:primes}
\begin{split}
&\left(s_{n+1}-s_{n}-6\right) \left(s_{n+1}-s_{n}-8\right) \left(s_{n+1}-s_{n}-10\right) \left(s_{n+1}-s_{n}-12\right)\\
&\Big(-4413259179869935469s_{n+1}^{2} s_{n}^{2}+9018202756736673077s_{n+1} s_{n}^{3}\\
&-4605857486207652970 s_{n}^{4}+17613028825077611962910s_{n+1}^{3}\\
&-49894515590732679680380s_{n+1}^{2} s_{n}+46888011363592541300913s_{n+1} s_{n}^{2}\\
&-14608471402583429770869 s_{n}^{3}-1352963324244298950066469s_{n+1}^{2}\\
&+2681855784524737400584915s_{n+1} s_{n}-1327176208257158948463120s_{n}^{2}\\
&+21624567546730578358641529s_{n+1}-21970446172359255395471255 s_{n}\\
&-64805493224537500905120572\Big)=0.
\end{split}
\end{equation}
\end{small}

Let $(\p_n)_{n\in\NN}$ be the sequence of prime numbers. The presence of long integers makes the above ADE unlikely to be correct. However, a critical remark is that it explicitly involves the prime gaps $6,8,10,$ and $12$. To us, this is an indication that this sequence of every second prime is not D-algebraic. Note that based on \cite[Theorem 5.1]{teguia2024computing}, if this sequence is not D-algebraic, then the sequence of primes is also not D-algebraic. For otherwise $(\p_{2n+1})_{n\in\NN}$, which is a subsequence of $(\p_n)_{n\in\NN}$ indexed by an arithmetic progression, would be D-algebraic. 

The ADE in \eqref{eq:primes} can only predict $56$ terms accurately. With the first $100$ terms in the input, we find another first-order ADE, this time of degree $12$, with the prime gaps $6,8,10,$ $12,14,$ and $18$ appearing in linear factors as in \eqref{eq:primes}. Note that, according to Definition \ref{def:dalgseq}, if this sequence is D-algebraic with defining difference polynomial having factors of the form $s_{n+1}-s_n-g$, $g\in2\NN$, then the corresponding $g$ appears infinitely often as a prime gap, an unsolved conjecture for all even integers. Observe that $g$ must be even as all terms of $(\p_{2n+1})_{n\in\NN}$ are odd.

In 1849, de Polignac conjectured that any even integer occurs infinitely often as a difference of two primes. It is easy to show that no first-order difference polynomial can define $(\p_{2n+1})_{n\in\NN}$ if de Polignac's conjecture holds. Indeed, the defining difference polynomial would have an unbounded degree because infinitely many $s_{n+1}-s_n-g,$ $g\in 2\NN,$ would divide it. This is a consequence of the fact that for any polynomial $p\in\KK[x_1,x_2],$ if $p(x,x+g)=0$ for infinitely many $x$ and a fixed $g$, then $p$ must vanish on the line defined by $x_2=x_1+g.$

However, we do not need to rely on de Polignac's conjecture to prove that primes are not D-algebraic. The above reasoning may be generalized to difference polynomials of arbitrary order based on existing results. There have been significant new results on prime gaps over the last two decades. The reference \cite{Granville2015} provides details of some of the main breakthroughs, which we exploit to prove that $(\p_{2n+1})_{n\in\NN}$ is not D-algebraic. The proof idea was explained to us by Florian Luca, who provided us with a complete direct proof of the non-D-algebraicity of the sequence of primes.

Before stating the theorem and its proof, we give some recalls. A $k$-tuple $(g_1,\ldots,g_k)$ is admissible if for every prime $\p$ there is an integer $n$ such that $(n+g_1)\cdots(n+g_k)\not\equiv 0\!\!\! \mod \p.$ Given $k$, it is straightforward to create admissible $k$-tuples. One can take any $k$ distinct integers, all multiples of $k!$. Maynard and Tao (2013, see details and references in \cite{Granville2015}) showed that for any $r\ge 2,$ there is a number $M(r)$ such that from any admissible $k$-tuple $(g_1,\ldots,g_k),$ with $k> M(r),$ there is $I\subset \{1,\ldots,k\}$ of cardinality $r$ such that $n+g_i, i\in I,$ are all primes for infinitely many $n$. Any integer $M$ such that $M \log M> e^{8r+4}$ can play the role of $M(r).$ In \cite{BanksFreibergTurnageButterbaugh2015}, it was shown that one can further ensure that these primes are consecutive in the sequence of primes.

\begin{theorem}\label{th:primes} The sequence $(\p_{2n+1})_{n\in\NN}$ is not D-algebraic.
\end{theorem}

\begin{proof} Let $(g_1,\ldots,g_k)$ be an admissible $k$-tuple, with $k$ sufficiently large to apply Maynard-Tao's theorem. Let $n$ be such that $\p_{j_i}=n+g_{j_i},$ with $j_i\!\!\in\!\!\!\{j_1,\ldots,j_r\}$, are consecutive terms in $(\p_{2n+1})_{n\in\NN}.$ This can be ensured using the results in \cite{BanksFreibergTurnageButterbaugh2015} by taking $r'=2r$ and every odd-indexed prime among the consecutive primes.

Thus, if $(\p_{2n+1})_{n\in\NN}$ is D-algebraic of order $r$ with defining polynomial $p$, we have
\[0=p(\p_{j_1},\ldots,\p_{j_r})=p(n+g_{j_1},\ldots,n+g_{j_r}).\]
From the last equality, we can write $p_{\boldsymbol{g}}(n)=0,$ $\boldsymbol{g}=(g_{j_1},\ldots,g_{j_r}),$ where $p_{\boldsymbol{g}}\in \KK[x]$ (we may assume that $\KK=\QQ$). Thus, for infinitely many $n$ we have $p_{\boldsymbol{g}}(n)=0$, which implies that $p_{\boldsymbol{g}}=0$. This, however, is impossible for all choices of admissible $k$-tuples.
\end{proof}
\noindent Regarding the concluding statement in the proof above, it is not difficult to construct $k$-tuples that do not vanish multivariate polynomials of $k$ variables. The main idea is to exploit the fact that the components of these $k$-tuples can be arbitrarily large.

\medskip

We hope our algorithms and their applications highlight the interest in studying D-algebraic objects and why guessing can sometimes be an effective alternative for computing and exploring with them. We were unable to compute most minimal differential polynomials for the challenging dynamical systems in \cite{mukhina2025projecting,mukhina2025support} within a reasonable time, likely due to the slow instantiation of multisums from recurrence equations -- a direction for future work. Nonetheless, our current algorithms and implementation remain practical for computing differential and difference equations of arbitrary degree and order, provided the expected output size is modest.

\section*{Acknowledgments}
\noindent This work was supported by UKRI Frontier Research Grant EP/X033813/1. The author is indebted to Florian Luca for answering all his questions about prime numbers. Small parts of this paper were presented at ACA 2025, in the sessions AADIOS and D-finiteness and beyond. The author thanks the participants for stimulating discussions.


\end{document}